\documentclass{easychair}

\usepackage{doc}
\usepackage{makeidx}
\usepackage[ruled, linesnumbered, vlined]{algorithm2e}
\usepackage[normalem]{ulem}

\usepackage{xcolor}
\usepackage{amsmath}
\usepackage{amssymb}
\usepackage{amsthm}
\usepackage{semantic}

\usepackage[utf8]{inputenc}
\usepackage{multicol}
\usepackage{array}
\usepackage{caption}
\usepackage{subcaption}

\newcommand{\n}[1]{\ensuremath{\mathit{#1}}}

\newcommand{\fofalse}{{\ensuremath{{\mathbf{false}}}}}
\newcommand{\foand}{\ensuremath{\land}}
\newcommand{\foor}{\ensuremath{\lor}}

\newcommand{\unary}[1]{\texttt{<}{#1}\texttt{>}}

\newcommand{\scs}{\mathcal{S}}
\newcommand{\sa}{{\mathcal{S}_{\n{A}}}}
\newcommand{\mc}[1]{\mathcal{#1}}

\newcommand{\eqih}{{\buildrel\rm i.h.\over=}}
\newcommand{\vgt}[1]{\n{vGT}(#1)}
\newcommand{\vgts}[1]{\n{vGT}(#1)_{\scs_A}}
\newcommand{\fgt}[2]{\n{fGT}(#1, #2)_{\scs_A}}
\DeclareMathOperator*{\dotin}{\dot{\in}}
\DeclareMathOperator*{\dotsubset}{\dot{\subseteq}}
\DeclareMathOperator*{\ndotsubset}{\dot{\nsubseteq}}

\makeatletter
\newtheorem*{rep@theorem}{\rep@title}
\newcommand{\newreptheorem}[2]{%
\newenvironment{rep#1}[1]{%
 \def\rep@title{#2 \ref{##1}}%
 \begin{rep@theorem}}%
 {\end{rep@theorem}}}
\makeatother

\newtheorem{example}{Example}
\newtheorem{definition}{Definition}
\newtheorem{proposition}{Proposition}
\newtheorem{corollary}{Corollary}
\newreptheorem{corollary}{Corollary}
\newtheorem{lemma}{Lemma}
\newreptheorem{lemma}{Lemma}
\newtheorem{theorem}{Theorem}

\begin{document}

%
\title{Reducing the Complexity of Quantified Formulas via Variable Elimination}


\titlerunning{}

%
\author{
 Aboubakr Achraf El Ghazi 
\and
 Mattias Ulbrich
\and
 Mana Taghdiri 
 \and 
 Mihai Herda
}

\institute{
 Karlsruhe Institute of Technology, Germany \\
 \email{$\{$elghazi, ulbrich, mana.taghdiri$\}$@kit.edu, mihai.herda@student.kit.edu} 
}


\authorrunning{A. A. El Ghazi and M. Taghdiri}

\clearpage

\maketitle

\begin{abstract}
We present a general simplification of quantified SMT formulas using variable elimination. 
The simplification is based on an analysis of the ground terms occurring as arguments in function applications.
We use this information to generate a system of set constraints, which is then solved to compute a set of {\em sufficient ground terms} for each variable. 
Universally quantified variables with a finite set of sufficient ground terms can be eliminated by instantiating them with the computed ground terms. 
The resulting SMT formula contains potentially fewer quantifiers and thus is potentially easier to solve. 
We describe how a satisfying model of the resulting formula can be modified to satisfy the original formula. 
Our experiments  show that in many cases, this simplification considerably improves the solving time, and our evaluations using Z3 \cite{z3} and CVC4 \cite{cvc4} indicate that the idea is not specific to a particular solver, but can be applied in general.
\end{abstract}


%
%

\pagestyle{empty}

\section{Introduction}
\label{sect:introduction}

Determining the satisfiability of first-order formulas with respect to theories is of central importance for system specification and verification. 
Current Satisfiability Modulo Theories (SMT) solvers have made significant progress in handling this problem efficiently.
SMT solvers such as CVC4 \cite{cvc4}, Yices1 \cite{yices1}, and Z3 \cite{z3} 
successfully address formulas containing quantifiers. 
They solve quantified formulas using heuristic quantifier instantiation based on the E-matching instantiation algorithm which was first introduced by Simplify \cite{simplify}. 
Although E-matching, because of its heuristic nature, is not complete, not even refutationally, it is best suited for integration into the DPLL(T) framework. 
Some techniques (e.g. \cite{e-matching_free-var, mbqi}) have extended E-matching in order to make it complete for some fragments of first-order logic. 

In spite of all the advances, the presence of quantifiers still poses a challenge to the solvers. 
In this paper, we propose a simplification of quantified SMT formulas that can be applied as a pre-process before calling an SMT solver. Given a (skolemized) SMT formula $A$, our simplification returns an equisatisfiable SMT formula $A'$ with potentially fewer universally quantified variables. Our simplification approach is syntactic in the sense that it extracts a set of set-valued constraints from the structure of $A$ whose solution is a set of {\em sufficient ground terms} for every variable. Those variables whose sets of sufficient ground terms are finite can be eliminated by instantiating them with the computed ground terms. If the resulting formula $A'$ is unsatisfiable, $A$ is guaranteed to be unsatisfiable too. However, if $A'$ has a model, it is not necessarily a model of $A$. We describe how any  model of $A'$ can be modified into a  model for $A$ without any significant overhead. This requires a special treatment of the interpreted functions. Our simplification procedure can also be applied if the 
logic of the input formula is not decidable; it can still reduce the number of quantifiers, thus simplifying the proof obligation.

Although our elimination process reduces the number of quantifiers, it may increase the number of occurrences of the remaining quantified variables (if any) (Appendix \ref{example} gives an example). Depending on the complexity of the involved terms, this may introduce additional overhead for the solver. Therefore, in order to apply our simplification as a general preprocessing step, it is important to balance the number of eliminated variables and the number of newly introduced variable occurrences. We define a metric that aims for estimating the cost of variable elimination, and allow the user to provide a threshold for the estimated cost. 

We have applied our simplification approach to 201 benchmarks from the SMT competition 2012 using CVC4 and Z3. The results indicate that in many cases, this simplification significantly improves the solving time, especially when a cost threshold is applied.
\vspace{-1mm}
\section{Background}
\label{sect:background}

This section provides a background on the first-order logic (FOL) (see \cite{logic-for-cs} for more details).
{\em Terms} are constructed from variables in $Var$, predicate symbols in
$P$ and function symbols in $F$\footnote{We distinguish between functions and predicates only when needed.}. 
Predicate and function symbols are given an arity by $\alpha: F \cup P \to \mathbb{N}$.
Function symbols with arity 0 are called constants and are denoted as $Con \subseteq
F$.  The set $Term$ of terms and the set $For$ of formulas are
defined inductively as usual. Terms without variables are called
\emph{ground terms} and denoted as $Gr \subseteq Term$. The set
$Gr(t)$ denotes all the ground terms occurring as subterms in a term $t$. We write $t[x_{1:n}]$ to denote that the variables $x_1, \ldots, x_n$ (for short $x_{1:n}$) occur in a
term $t$.
For an expression $t \in Term \cup For$, a variable $x$ and a ground term $gt$, 
the expression $t[gt/x]$ substitutes $gt$ for all the occurrences of $x$ in $t$. 
We apply substitutions (aka. instantiations) also to finite sets $S$ of ground terms as $t[S/x] := \{ t[gt/x] \mid gt \in S\}$. 
The Herbrand universe $\mc{H}(A)$ of a formula $A$ is the set of all ground terms built from $A$. 
That is, all constants occurring in $A$, are in $\mc{H}(A)$, and 
for each function $f$ occurring in $A$ and $gt_1, \dots, gt_{\alpha(f)} \in \mc{H}(A)$, $f(gt_1, \dots, gt_{\alpha(f)}) \in \mc{H}(A)$.

A literal is an atomic formula or a negated atomic formula. 
A clause is a disjunction of literals.  
A formula is in \emph{clause normal form} (CNF) if it is a conjunction $(C_1 \foand \dots \foand C_n)$ of
clauses where all $C_i$ are quantifier-free 
and all variables are implicitly universally quantified.  
We assume, unless stated otherwise, that all considered formulas are in CNF and all variables are unique. 
When required, we refer to clauses and CNFs as sets of literals and clauses, respectively.  

A semantical \emph{structure} (also called a \emph{model}) $\mc{M}$
is a tuple $(|M|, M)$, with a non-empty universe $|M|$, and a mapping
$M$ that defines an {\em interpretation} for every symbol in $F \cup P$, i.e.  
for $f \in F$, $M(f): |M|^{\alpha(f)} \rightarrow |M|$, and 
for $p \in P$, $M(p) \subseteq |M|^{\alpha(p)}$. 
Variables get their values from a variable assignment function $\beta: Var \rightarrow |M|$. 
The interpretation $(M,\beta)(t)$ of a term $t$ is defined inductively, and the interpretation of a set of terms $S$ is defined as
$(M,\beta)(S) = \{(M,\beta)(s) \mid s \in S\}$.
For a formula $A \in For$, we use ${\cal M} \models A$ if ${\cal M}$ is a satisfying model (or, for short, a model) of $A$, i.e. $A$ is true in ${\cal M}$. 
We use $\models A$ if $A$ is universally valid.

A \emph{theory} $\mc{T}$ is a deductively closed  set of formulas.  
A $\mc{T}$-model $\mc{M}$ is a model that satisfies all the formulas in $\mc{T}$.
A formula $A\in For$ is satisfiable modulo theory ${\cal T}$ if there 
exists a $\mc{T}$-model with
${\cal M} \models A$, for short $\mc{M} \models_{\mc{T}} A$.  
The function symbols that have their semantics (partially) fixed by 
${\cal T}$ are called {\em interpreted} and all others are \emph{uninterpreted}. 
If a term contains an interpreted function which is applied to a variable, we call it an {\em interpreted term}, otherwise, an {\em uninterpreted term}.
%
We denote variables by $x, y, \dots$; constants by $a, b, \dots$;  ground terms by $gt_i$; 
uninterpreted functions by $f, g, \dots$; interpreted functions by $op_i$; predicates  by $p, q, \dots$; 
terms by $s, t, \dots$; formulas by $A, B, \dots$;  values by $v_i$; and the considered SMT theory by $\mc{T}$. 


\vspace{-2mm}
\section{Example}

Figure \ref{examp2}(a) shows an SMT formula (as a set of implicitly conjoined subformulas) in which $c_1$ and $c_2$ represent constants, $f$ is a unary function, and $p$ is a binary predicate. 
Figure \ref{examp2}(b) shows the same formula after conversion to CNF: constants $c_3$ and $c_4$ denote the skolems for the formulas (3) and (4), respectively. 
Instead of solving the original formula (denoted by $A$), we produce an {\em instantiated formula} $A^{inst}$ in which the $x$ and $y$ variables are instantiated with certain ground terms.
$A^{inst}$ is given in Figure \ref{examp2}(c) where the numbers correspond to the lines in the CNF (and original) formula.  Formula $A^{inst}$ has fewer quantifiers than $A$ (in fact, it has zero quantifiers), and thus is easier to solve. 
We use $\n{vGT}(x)$ to represent the set of ground terms that is used to instantiate a variable $x$.
Variable $x$ (in Formula 2) refers to the first argument of $f$, and thus we instantiate it with all the ground terms that occur in that position, namely $\{c_1, c_4\}$. We call this the set of ground terms of $f$ for argument position 1, and denote it by $\n{fGT}(f, 1)$. Variable $y$ (in Formula 3), on the other hand, refers to both the first argument of $p$ and the first argument of $f$. Therefore, $\n{vGT}(y) = \n{fGT}(p, 1) \cup \n{fGT}(f, 1)$. In order to guarantee equisatisfiability of $A^{inst}$ and $A$, if two functions are applied to the same variable, they should be instantiated with the ground terms of both functions (see Section \ref{sect:sufGTSets}). Therefore, in this example, $\n{fGT}(p, 1) = \n{fGT}(f, 1) = \{c_1, c_4\}$ although $p$ is not directly applied to any constants.

\begin{figure}[t]
\begin{tabular}{c|c|c}
\begin{minipage}{4.9cm}
$(1)\ c_1 \neq c_2 \\
 (2)\ \forall x \mid f(x) = f(c_1) \\
 (3)\ \exists z \mid \forall y \mid \neg p(y, z) \foor f(y) = c_2 \\
 (4)\ \exists z \mid f(z) = c_1 
$\end{minipage}&

\begin{minipage}{4.3cm}
$(1)\ c_1 \neq c_2 \\
 (2)\ \forall x \mid f(x) = f(c_1) \\
 (3)\ \forall y \mid \neg p(y, c_3) \foor f(y) = c_2 \\
 (4)\ f(c_4) = c_1 
$\end{minipage}&

\begin{minipage}{3.8cm}
$(1)\ c_1 \neq c_2 \\
 (2)\ f(c_1) = f(c_1) \\
 (2)\ f(c_4) = f(c_1) \\
 (3)\ \neg p(c_1, c_3) \foor f(c_1) = c_2 \\
 (3)\ \neg p(c_4, c_3) \foor f(c_4) = c_2 \\
 (4)\ f(c_4) = c_1 
$\end{minipage}\\ 
& & \\
(a) & (b) & (c)\\
\end{tabular}

\begin{minipage}{13cm}
$\\ 
M(c_1) = 1, M(c_2) = 2, M(c_3) = 3, M(c_4) = 4\\
 M(f)(v) = 
\begin{cases}
1 & \n{if} \ v = 1 \\
1 & \n{if} \ v = 4 \\
any\ value & \n{else} 
\end{cases} 
\ \ \ \ \ \ \ \ \ M(p)(v, 3) = 
\begin{cases}
\fofalse & \n{if} \ v = 1 \\
\fofalse & \n{if} \ v = 4 \\
any\ value & \n{else} 
\end{cases}
$\end{minipage}\\ \\
\verb|                                      |(d)

\begin{minipage}{15cm}
$\\
M^\pi(c_1) = M(c_1) = 1, M^\pi(c_2) = M(c_2) = 2, M^\pi(c_3) = M(c_3) = 3, M^\pi(c_4) = M(c_4) = 4\\
M^{\pi}(f)(v) = 
\begin{cases}
M(f)(v)  & \n{if} \ v \in \{1, 4\} \\
M(f)(M(c_1)) & \n{else} 
\end{cases}\ \ \ \ \ \ \ \ \ \ \ \ \ \ \ \ \ = 1\ \ \ \  \n{for\ all\ v}\\
M^{\pi}(p)(v, c_3) = 
\begin{cases}
M(p)(v, M(c_3)) & \n{if} \ v \in \{1, 4\} \\
M(p)(M(c_1), M(c_3)) & \n{else} 
\end{cases}\ \ \ \  = \fofalse\ \ \ \  \n{for\ all\ v}
$\end{minipage}\\ \\
\verb|                                      |(e)
\caption{Example. (a) original SMT formula, (b) CNF formula, (c) instantiated formula, (d) a model for the instantiated formula, and (e) a model for the original formula.}
\label{examp2}
\end{figure}

The instantiated formula is an implication of the original formula. Hence, if $A^{inst}$ is unsatisfiable,  $A$ is also unsatisfiable. However, not every  model of $A^{inst}$ satisfies $A$. But the instantiation was chosen in such a way that we can modify the models of $A^{inst}$ to satisfy $A$. Figure \ref{examp2}(d) gives a sample  model $\mc{M}$ for $A^{inst}$ which does not satisfy $A$. Since in $A^{inst}$, $f$ is only applied to $c_1$ and $c_4$, and $p$ only to $(c_1, c_3)$ and $(c_4, c_3)$,  $\mc{M}$ may assign arbitrary values to $f$ and $p$ applied to other arguments. Although these values do not affect  satisfiability of $A^{inst}$, they affect  satisfiability of $A$. Therefore, we modify $\mc{M}$ to a model $\mc{M}^\pi$ by defining acceptable values for the function applications that do not occur in $A^{inst}$. Figure \ref{examp2}(e) gives the modified model $\mc{M}^\pi$ that our algorithm constructs. It is easy to show that this model  satisfies $A$.

The basic idea of modifying a model is to fix the values of the function applications that do not occur in $A^{inst}$ to some arbitrary value of a function application that does occur in $A^{inst}$.
This  works well for this example as $f$ and $g$ are uninterpreted symbols and thus their interpretations are not restricted beyond the input formula. 
Were they interpreted symbols, this would be different.  As an example, assume
that $p$  is the interpreted operator
``$\leq$''. In this case, the original formula $A_{\leq}$ becomes
unsatisfiable\footnote{(2) and (4) imply $f(c_1) = c_1$. $y \le z$ holds for some pair of integers, thus (3) implies $f(y) = c_2$ for some $y$. But $f(y) = f(c_1)$ by (2) and so $f(c_1) = c_2 = c_1$. This contradicts (1).}, 
but its instantiation ${A}_\leq^{\n{inst}}$ stays satisfiable\footnote{A model is $M'(c_1) = 1, M'(c_2) = 2, M'(c_3) = 0, M'(c_4) = 4, M'(f) \equiv 1$}. 
To guarantee the equisatisfiability in the presence of interpreted literals, we require the ground term sets to contain some terms that make the interpreted literals false. This makes the solver explore the cases where clauses become satisfiable regardless of the interpreted literals. 
In this example, the interpreted literal $\neg (y \le c_3)$ becomes false if $y$ is instantiated with the ground term $c_3-1$. Instantiating $A_{\le}$ with the ground terms $\{c_1, c_4, c_3-1\}$ reveals the unsatisfiability.


\section{Sufficient Ground Term Sets}
\label{sect:sufGTSets}


\begin{definition}
\label{def1}
Given a variable $x$ in an SMT formula $A$ (in CNF), a set of ground terms $S \subseteq \mc{H}(A)$
is  \emph{sufficient} for $x$ w.r.t a theory ${\cal T}$ if $A$ and $A[\n{S/x}]$ are equisatisfiable modulo ${\cal T}$.
\end{definition}

A variable $x$ in a formula $A$ can have more than one sufficient set of ground
terms. $\mc{H}(A)$ is always a sufficient set 
of ground terms as a result of the G\"odel-Herbrand-Skolem theorem which states 
that a formula $A$ in Skolem Normal Form (SNF) is satisfiable iff
$A[\mc{H}(A)/x]$ is satisfiable \cite{logic-for-cs}. But $\mc{H}(A)$ is usually infinite, and
our goal is to determine whether a \emph{finite} set of
sufficient ground terms exists, and to compute it if one exists.  This
computation is done by generating and solving a system of set
constraints over  sets of ground terms. 

Figure \ref{fig:sufGT-CRules} presents our (syntactic) rules 
to generate the set constraints for a formula $A$ in CNF.  
The notation $t \dotin C$ denotes that a term $t$ occurs 
as a subterm of a clause $C$.
We use $\scs_{\n{A}}$ to denote the set constraints system that results from applying these rules exhaustively to all the clauses of $A$. 
The constraints range over the sets $\vgt{x} \subseteq Gr$ for all variables $x$ in A. 
These sets denote the relevant instantiations for the respective variables. Auxiliary sets
$\n{fGT}(f,i) \subseteq Gr$ are introduced to denote the set of
relevant ground terms for an uninterpreted function $f \in F$ at an
argument position $i \in \mathbb{N}$.
We assume that the theory of integers is part of the considered $\mc{T}$, and that integers are included in the universe of every $\mc{T}$-model $\mc{M}$, i.e. $\mathbb{Z} \subseteq |M|$. 
The integer operators $<, \leq, +, -, \geq, >$ are fixed with their obvious meanings.

Rule $R_0$ of Figure \ref{fig:sufGT-CRules} guarantees that the set of relevant ground terms is not empty for any variable in $A$. 
Rule $R_1$ establishes a relationship between sets of ground terms for variables and function arguments.
Rule $R_2$ ensures that the ground terms that occur as arguments of a function $f$ are added to the corresponding ground term set of $f$.
Rule $R_3$ states that if a term $t[x_{1:n}]$ with variables $x_{1:n}$ occurs as the $i$-th argument of $f$, 
then all the instantiations of $t$ with the respective sets $\n{vGT}(x_i)$ must be in $\n{fGT}(f,i)$.
Rule $R_4$ states that our approach does not currently handle the case where a variable $x$ occurs as an argument of an unsupported interpreted function (supported operators are $\{=, <, \leq, >, \geq\}$), thus sets $\vgt{x}$ to infinity\footnote{In theory, this infinite set denotes $\mc{H}(A)$, but we use it as the ``unsupported" label that gets propagated to other relevant sets.} in order to be propagated to other relevant ground term sets.
Moreover, we do not handle the case where a supported interpreted operator has more than one variable argument (rule $R_5$). 
The remaining rules infer additional constraints for $\vgt{x}$ where $x$ occurs as an argument of a supported interpreted function. 
They constrain $\vgt{x}$ to contain at least one ground term that falsifies the corresponding (interpreted) literal.

\begin{figure}[tb]
\begin{math}
\inference[R$_0$:]{x \dotin C}{\vgt{x} \not = \emptyset} \hspace{1cm} 
\inference[R$_1$:]{f(\cdots, \overbrace{x}^{\text{i-th}}, \cdots) \dotin C}{\n{vGT(x)} = \n{fGT(f,i)}}  \hspace{1cm} 
\inference[R$_2$:]{f(\cdots, \overbrace{gt}^{\text{i-th}}, \cdots) \dotin C}{gt \in \n{fGT(f,i)}} \\ \vspace{.2cm}
\inference[R$_3$:]{f(\cdots, \overbrace{t[x_{1:n}]}^{\text{i-th}}, \cdots) \dotin C}{t[\n{vGT}(x_1)/x_1, \cdots, \n{vGT}(x_n)/x_n] \subseteq \n{fGT(f,i)}}\\ \vspace{.2cm}
\inference[R$_4$:]{op(\cdots, x, \cdots) \in C, \ op \not\in \{=, <, \leq, >, \geq\}}{\vgt{x} = \infty} \hspace{1.3cm}
\inference[R$_5$:]{op(x, y) \in C,  \ op \in \{=, <, \leq, >, \geq\}}{\n{vGT}(x) = \infty\ \ \n{vGT}(y) = \infty} \\ \vspace{.2cm}
\inference[R$_6$:]{(x \leq gt) \in C}{gt+1 \in \n{vGT(x)}} \hspace{0.4cm}
\inference[R$_7$:]{(x \geq gt) \in C}{gt-1 \in \n{vGT(x)}} \hspace{0.4cm}
\inference[R$_8$:]{\neg op(x, gt) \in C, \n{where} \ op \in \{\leq, \geq\} }{gt \in \n{vGT(x)}} \\ \vspace{.2cm}
\inference[R$_9$:]{\neg (x < gt) \in C}{gt-1 \in \n{vGT(x)}} \hspace{0.4cm} 
\inference[R$_{10}$:]{\neg (x > gt) \in C}{gt+1 \in \n{vGT(x)}}  \hspace{0.7cm} 
\inference[R$_{11}$:]{\n{op}(x, gt) \in C, \n{where \ \n{op} \in \{<, >\}}}{gt \in \n{vGT(x)}} \\ \vspace{.2cm}
\inference[R$_{12}$:]{\neg (x = gt) \in C}{gt \in \n{vGT(x)}} \hspace{0.9cm} 
\inference[R$_{13}$:]{(x = gt) \in C, x \in \mathbb{Z}}{\{gt-1, gt+1\} \subseteq \n{vGT(x)}} \hspace{0.7cm}
\inference[R$_{14}$:]{(x = gt) \in C, x \notin \mathbb{Z}}{\n{vGT(x)} = \infty}
\end{math}
\caption{The syntactic rules for generating the set constraints system ($\scs_A$).} 
\label{fig:sufGT-CRules}
\end{figure}

Let $vGT_{\scs_A}$ denote a collection of finite sets of ground terms which satisfies 
the constraints $\scs_A$.
We show that, if finite, $\n{vGT}(x)_{\scs_A}$ is a sufficient ground term set 
for $x$ in $A$. The variable $x$
can hence be eliminated by instantiating it with all the ground terms in
$\n{vGT}(x)_{\scs_A}$. The resulting formula $A[\n{vGT}(x)_{\scs_A}/x]$ is
equisatisfiable to $A$ and does not contain $x$ anymore.


\begin{theorem}[Main Theorem]
\label{theo:1}
Let $x$ be a variable in $A$ with $\n{vGT}(x)_{\scs_A} \neq \infty$, then $A$ and $A[\n{vGT(x)_{\scs_A}/x}]$ are equisatisfiable.
\end{theorem}

\begin{proof}
If $A[\n{vGT(x)_{\scs_A}/x}]$ is unsatisfiable, so is $A$ since the former is an implication of the latter. 
If $A[\n{vGT(x)_{\scs_A}/x}]$ is satisfiable with a model $\mc{M}$, then we construct a modified model $\mc{M}^{\pi_x}$ (as defined below) and show in lemma \ref{lemma:Mpi-sat-A} that $\mc{M}^{\pi_x}$ satisfies $A$. 
\end{proof}

Given a model $\mc{M}$ for the formula $A[\n{vGT(x)_{\scs_A}}/x]$, we construct a modified model $\mc{M}^{\pi_x}$ as follows: $|M^{\pi_x}| := |M|$. 
For any constant $c \in Con$, $M^{\pi_x}(c) := M(c)$. For any interpreted operator $op$, $M^{\pi_x}(op) := M(op)$. 
For any uninterpreted function $f$, $M^{\pi_x}(f)(v_{1:n}) := M(f)(\pi_x(f,1)(v_1), \cdots, \pi_x(f,n)(v_n))$, where 
$\pi_x(f,i)$ is defined as in Eq. \ref{eq:pi-fi}.
Intuitively, 
if the ground term set of $x$ does not {\em subsume} the ground term set of the $i^{th}$ argument of $f$,
or if $v_i$ is a value that $M$ assigns to a ground term for the $i^{th}$ argument of $f$, then $M^{\pi_x}(f)(.., v_i, ..) := M(f)(.., v_i, ..)$ 
Otherwise, $\pi_x(f,i)$ maps $v_i$ to a value that $M$ assigns to some ground term for the $i^{th}$ argument of $f$. 
Integers must be mapped to the closest such value (see the proof of Lemma \ref{lemma:1}). 
A ground term set $S$ {\em subsumes} a ground term set $R$, denoted by $R \dotsubset S$, 
if for every ground term $gt_1 \in R$ there exists a ground term $gt_2 \in S$ such that $gt_1$ is a subterm of $gt_2$.


\begin{gather}
\label{eq:pi-fi}
%
\pi_x(f,i)(v) = 
\begin{cases}
  v & \n{if} \ \fgt{f}{i} \ndotsubset \vgts{x} \\
  v & \n{else \ if} \ v \in M(\n{fGT}(f,i)_{\scs_A})\\
  v' \in M(\n{fGT}(f,i)_{\scs_A}) & \n{else \ if} \ v \not\in \mathbb{Z} \\
  v'\!\in\!M(\n{fGT}(f,i)_{\scs_A}), \n{s.t.} \ |v-v'| \ \n{is\ minimal} & \n{otherwise}
\end{cases}
\end{gather}

\begin{gather}
\label{eq:pi-x}
\pi_x(v) =  
\begin{cases}
v & \n{if} \ v \in M(\n{vGT}(x)_{\scs_A}) \\
v' \in M(\n{vGT}(x)_{\scs_A}) &  \n{else \ if}\ v \notin \mathbb{Z} \\
v' \in M(\n{vGT}(x)_{\scs_A}), \n{s.t.} \ |v-v'| \ \n{is\ minimal} & \n{otherwise}
\end{cases}
\end{gather}

We also define $\pi_x$ (as in Eq. \ref{eq:pi-x}) to denote the value projection with respect to a variable $x$. 
If $\vgts{x} = \fgt{f}{i}$, for instance because $x$ occurs as the $i^{th}$ argument of $f$, then $\pi_x = \pi_x(f,i)$. 
%
Before showing the proof of lemma \ref{lemma:Mpi-sat-A} used in our main theorem, we introduce some auxiliary corollaries and lemmas. The proofs of the lemmas can be found in  Appendix \ref{a:sect:proofs}.

%
\begin{corollary}
\label{col:1}
If $\vgts{x} \neq \infty$, then $\pi_x(v) \in M(\vgts{x})$, for all $v \in |M|$.
\end{corollary}


The following lemmas show that if $\mc M^{\pi_x}$ does not satisfy a
literal $l$ in a CNF formula $A$, a modified variable assignment $\beta'$ can be
found such that $\mc M$ together with $\beta'$ does not satisfy $l$. 
Lemma~\ref{lemma:1} formulates the claim for interpreted literals, 
and Lemma~\ref{lemma:3} gives a stronger variant (with value equality rather than implication) 
for uninterpreted literals.

%
\begin{lemma}
\label{lemma:1}
Let $x$ be a variable with $\vgts{x} \neq \infty$, $\mc{M}$ a model, $\beta$ a variable assignment, and 
$\beta' = \lambda y.\ \n{if}\ \vgts{y} \dotsubset \vgts{x} \ \n{then}\ \pi_y(\beta(y))\ \n{else}\ \beta(y)$.
Then  
$(M, \beta') \models l$ implies $(M^{\pi_x}, \beta) \models l$
for all interpreted literals $l$ in $A$.
\end{lemma}

\begin{lemma}
\label{lemma:3}
Let $x$ be a variable with $\vgts{x} \neq \infty$, $\mc{M}$ a model, $\beta$ a variable assignment, and 
$\beta' = \lambda y.\ \n{if}\ \vgts{y} \dotsubset \vgts{x} \ \n{then}\ \pi_y(\beta(y))\ \n{else}\ \beta(y)$.
Then 
$
(M,\beta')(l)
= 
(M^{\pi_x}, \beta)(l) 
$ for all uninterpreted literals $l$ in $A$.
\end{lemma}

\begin{lemma}
\label{lemma:Mpi-sat-A}
Let $x$ be a variable in $A$ with $\n{vGT}(x)_{\scs_A} \neq \infty$ and $\mc{M}$ a model of $A[\n{vGT}(x)_{\scs_A}/x]$, then $\mc{M}^{\pi_x}$ is a model of $A$.
\end{lemma}
\begin{proof}
Let $A'$ denote $A[\n{vGT}(x)_{\scs_A}/x]$. Since $\mc{M}$ is a model of $A'$, for every variable assignment $\beta:Var \to |M|$, we have $(M, \beta) \models A'$.
Let $\beta_0$ be an arbitrary variable assignment. 
By corollary~\ref{col:1}, we know that $\pi_x(\beta_0(x)) = M(gt_0)$ for
some ground term $gt_0 \in \n{vGT}(x)_{\scs_A}$. The instantiation
$A[gt_0/x]$ is included in $A'$ and thus $(M,\beta) \models A[gt_0/x]$ for any $\beta$.
Let 
$\beta'_0 = \lambda y.\ \n{if}\ \vgts{y} \dotsubset \vgts{x} \ \n{then}\ \pi_y(\beta_0(y))\ \n{else}\ \beta_0(y)$.
Assignment $\beta'_0$ maps $x$ to $\pi_x(\beta_0(x)) = M(gt_0)$ and $(M,\beta'_0) \models A[gt_0/x]$, therefore $(M, \beta'_0) \models A$. \\
Assuming that $A$ is in CNF, 
there must be for every clause $C$ in $A$ a literal $l^C$ in $C$ with $(M,\beta'_0) \models l^C$.
Using lemma~\ref{lemma:1} for interpreted and lemma~\ref{lemma:3} for
uninterpreted literals, we know that also $(M^{\pi_x}, \beta_0) \models l^C$. 
Hence, $M^{\pi_x}$ is a model for $l^C$, $C$ and finally for $A$.

\end{proof}

%


\section{Practical Optimizations}
\label{sect:optimization}

%
\subsection{Simulating NNF}

Previous section established that if the input formula is in CNF, we can instantiate variables with their computed sets of sufficient ground terms. 
Computing such sets, however, does not require the formula to be in CNF. That is, the constraint system of Figure \ref{fig:sufGT-CRules}   
needs only the CNF polarity of the literals of the input formula (see rules $R_6$ to $R_{13}$).
Therefore, instead of actually converting the original formula to CNF, we (1) {\em simulate} the NNF (negation normal form) conversion (without actually changing the formula) to compute polarity, and (2) skolemize all existential quantifiers\footnote{If a formula $A$ is not in CNF, the instantiation of a variable $x$ with a set $S$ of ground terms should be adjusted as $A[S/x] := A[\bigwedge \limits_{gt \in S} B_x[gt/x] / B_x]$, where $B_x$ is the smallest subformula containing $x$.}.
This computation does not introduce any considerable overhead. 
It should be noted that conversion to CNF using distribution (as opposed to Tseitin encoding \cite{tseitin}) has the additional advantage that it minimizes the scope of each variable. 
This can significantly improve our simplification approach. 
Distribution, however, is very costly in practice. 
Computing minimal variable scopes without performing distribution is left for future work.

\subsection{Limiting Instantiations}
\label{limiting}

\begin{algorithm}[t]
\KwData{$A: For, C_{max}: \mathbb{N}$} 
\KwResult{$\n{NoElim}: Set\unary{Var}$}
\Begin{
 $\n{NoElim} \gets \{x \in vars(A) \mid \vgts{x} = \infty\}$ \\
 \Repeat{$\n{NoElim}\ is\ unchanged$}{
  \For{$x \in vars(A) \setminus \n{NoElim}$}{
    $\n{repFactor} \gets  |scopevars(x) \cap \n{NoElim}| = \emptyset \ ? \ 0 : 1$ \\
    $cost_x \gets (\prod\limits_{y \in scopevars(x)\setminus \n{NoElim}} |\vgts{y}|) * \n{repFactor}$ \\
    \If{$cost_x > C_{max}$}{
     $select\ m \in scopevars(x)\setminus \n{NoElim}\ s.t.\ |\vgts{m}|\ is\ maximum$\\
     $\n{NoElim} \gets \n{NoElim} \cup \{m\}$ 
    }
  }
 }
 \Return{$\n{NoElim}$}
}
\caption{Heuristic detection of expensive variables with respect to a threshold}
\label{algo:qbi2}
\end{algorithm}

Our simplification approach eliminates those variables that have finite sets of sufficient ground terms by instantiating them with the computed ground terms. 
In practice, such instantiation may increase the occurrences of non-eliminable variables (see the example of Appendix \ref{example}).
Our experiments with Z3 and CVC4 show that this increase in the number of variable occurrences can considerably increase the solving time, specially for nested quantifiers.


We use Algorithm~\ref{algo:qbi2} to estimate and limit the cost of variable elimination based on the number of  variable occurrences that it introduces. 
The algorithm tries to maximize  the number of eliminated variables while keeping the cost low. 
Given a formula $A$ and a threshold cost $C_{max}$, this algorithm returns a set of variables $\n{NoElim}$ whose elimination causes the cost to exceed $C_{max}$. 
Line $2$ initializes the $\n{NoElim}$ set to the set of all variables whose sets of sufficient ground terms are infinite, and thus will not be eliminated by our approach. 
Lines $4$-$9$ evaluate the cost of eliminating a variable $x$ that does not belong to $\n{NoElim}$. 
Instantiating $x$ with its sufficient ground terms, in the worst case, replicates all non-eliminable variables (either free or bound) 
that appear in the scope of $x$ (denoted by $\n{scopevars}(x)$), where the scope of $x$ is the body of the quantified formula that binds $x$. 
We estimate the cost of eliminating all eliminable variables in the scope of $x$ by $cost_x$. 
If this number exceeds the given threshold, then a variable $m$ with the maximum number of instantiations will be marked as non-eliminable. 
The process then starts over.
\section{Evaluation}
\label{sect:evaluation}

\begin{figure}[tpb]
 \centering
 \begin{subfigure}[b]{0.49\textwidth}
  \includegraphics[width=1\textwidth]{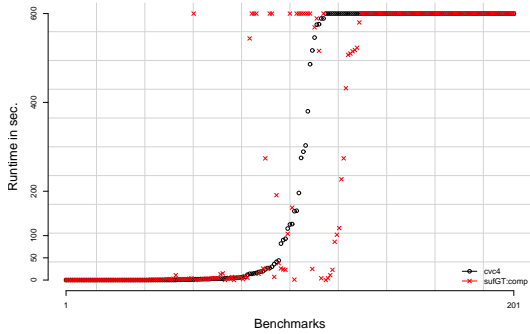}
  \caption{CVC4, original vs. simplified (complete)}
  \label{fig:cvc4-vs-sufgt:comp}
 \end{subfigure}
 \begin{subfigure}[b]{0.49\textwidth}
  \includegraphics[width=1\textwidth]{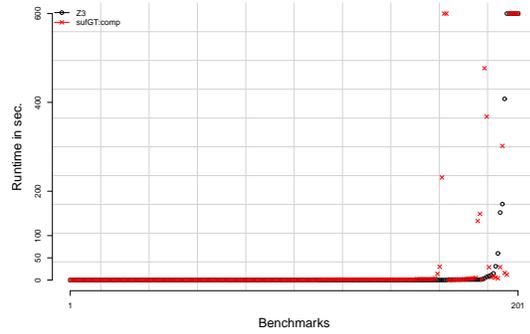}
  \caption{Z3, original vs. simplified (complete)}
  \label{fig:z3-vs-sufgt:comp}
 \end{subfigure}
 \begin{subfigure}[b]{0.49\textwidth}
  \includegraphics[width=1\textwidth]{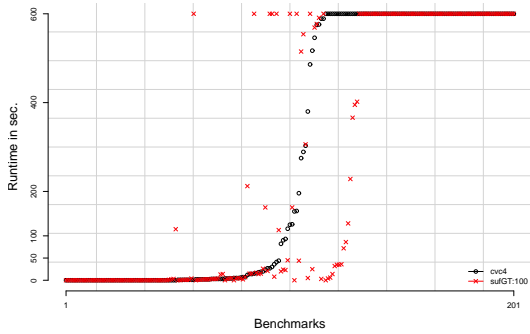}
  \caption{CVC4, original vs. simplified ($C_{max}$ = 100)}
  \label{fig:cvc4-vs-sufgt:20}
 \end{subfigure}
 \begin{subfigure}[b]{0.49\textwidth}
  \includegraphics[width=1\textwidth]{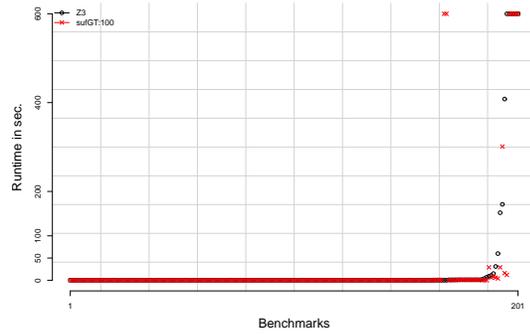}
  \caption{Z3, original vs. simplified ($C_{max}$ = 100)}
  \label{fig:z3-vs-sufgt:20}
 \end{subfigure}
\caption{Experimental results on the benchmarks of the SMT-COMP/AUFLIA-p}
\end{figure}

We have implemented our approach in a prototype tool and performed experiments on the  SMT-COMP benchmarks of 2012 in the AUFLIA-p/2012 division, using CVC4 (version 1.0) and Z3  (version 4.1) solvers.
We ran both solvers on all benchmarks on 
an AMD DualCore Opteron Quad, $2.6$GHz with $32$GB memory. 

For each benchmark, we compare the original runtime of each solver (with no simplification) against (1) a complete variable elimination, (2) a limited variable elimination where  $C_{max} = 100$. Figures \ref{fig:cvc4-vs-sufgt:comp} and \ref{fig:cvc4-vs-sufgt:20} give the comparison results for CVC4, 
and Figures \ref{fig:z3-vs-sufgt:comp} and \ref{fig:z3-vs-sufgt:20} give the results for Z3. The $x$-axis of each plot shows the benchmarks, sorted according to the original runtime of the solvers, and the $y$-axis gives the runtime in seconds. Time-outs and `unknown' outputs are represented identically. The time-out limit is 600 seconds.

For CVC4, the complete variable elimination improves the solving time of $37$  
cases  ($18$\%)--average speedup\footnote{Speedup = old solving time / new solving time, where 0 second is changed to $0.5$ second.}  $49$x--out of which $16$ 
were originally unsolvable, and worsens $55$  
cases ($27$\%)--average speedup $0.45$. 
The limited variable elimination, on the other hand, improves $39$ 
cases ($19$\%)--average speedup $57$x--out of which $15$  
were originally unsolvable, and worsens $32$ 
cases ($15$\%)--average speedup $0.48$. 
Z3 is known to be highly efficient in the AUFILA division (winner since 2008); its original runtime on many benchmarks is zero. 
The complete variable elimination, however, worsens $70$  
of these benchmarks ($34$\%)--average speedup $0.38$--and improves $11$ 
cases ($5$\%)--average speedup $10$x--out of which one was originally unsolvable. 
The limited variable elimination, on the other hand, worsens only $8$ 
cases ($4$\%)--average speedup $0.35$--and improves $14$  
cases ($7$\%)--average speedup $9.4$x--out of which one was  originally unsolvable. 
%

The main reason for slow down is the introduction of too many variable occurrences when not all variables are eliminable. Thus, as shown by these plots, 
for both solvers, the limited variable elimination produces  stronger results\footnote{Detailed information of the benchmarks are available at  {\tiny \url{http://i12www.ira.uka.de/~elghazi/sufGT_smt13_expData/}}}. 
However, even when {\em all} variables are eliminated, it is still possible that the solving time worsens as 
the number of instantiations that we produce can be higher than the number of instantiations that the solver would generate while solving the quantified formula. 
Although feasible in theory, this case was never observed in our experiments. 

Although variable elimination with a limited cost can result in significant improvements of solving time, the experiments show that in some cases such as the two new time-outs of Figure \ref{fig:z3-vs-sufgt:20}, a finer-grained limitation decision is needed. Investigating such cases is left as future work.
\vspace{-1mm}
\section{Related Work}
\label{sect:related-work}
Quantifier elimination in its traditional sense (aka. QE) refers to the property that an FOL theory $\mc{T}$ admits QE if for each formula $\phi$, there exists 
a quantifier-free formula $\phi'$ so that for all models $\mc{M}$, $\mc{M} \models_\mc{T} \phi \Leftrightarrow \phi'$.
Most  applications of QE either provide decision procedures for  fragments of FOL, or only prove their decidability. 
For example, the decidability proof of the Presburger arithmetic theory shows that the augmented theory with divisibility predicates admits QE \cite{enderton-2001}. 
Another example  is the Fourier-Motzkin QE procedure for linear rational arithmetic (see \cite{omega-test}).
 QE is applicable to formulas that are purely in one of the known arithmetic theories, and eliminates  those variables whose enclosing formulas are in a theory that admits QE. Consequently, 
it is not suitable as a general, stand-alone simplification for SMT formulas.

Another approach to eliminate quantifiers was proposed in \cite{gladisch-sat} where partial FOL models are represented as programs. A program generation technique tries to heuristically generate a program $P_i$ for a quantified formula $\phi_i$ in $F := \phi_1\wedge\ldots\wedge\phi_n$
such that the proof obligation $[P_i](\phi_1, \dots, \phi_n \Rightarrow \phi_i)$ can be discharged using a theorem prover. 
If such a program is found, $F$ is modified to $\phi'_1\wedge\ldots\wedge\phi'_n$ (without $\phi_i$) where $\phi'_j\equiv[P_i]\phi_j$. 
The program generation and verification loop can be repeated until all quantified formulas are eliminated.
Such an approach is very different from ours and is sound only for satisfiable formulas.


Our work was motivated by \cite{array-property-fragment} and \cite{mbqi} in which quantifiers are eliminated via instantiation. In \cite{array-property-fragment}, a decision procedure is proposed for the \emph{Array Property} fragment of FOL which supports a combination of Presburger arithmetic for  index terms, and equality with uninterpreted functions and sorts (EUF) for  array terms. Similar to ours, this work  
instantiates universally quantified variables with a finite set of ground terms to generate an equisatisfiable formula.
They prove the existence of such sets for their target fragment. Our approach, however, targets general FOL  and leaves a variable uninstantiated if its  set of ground terms is infinite. We believe that we can successfully handle the Array Property fragment. Experiments are left for future work.


In \cite{mbqi},  Model-based Quantifier Instantiation (MBQI) is proposed for Z3. Similar to ours, this work constructs a system of set constraints $\varDelta_F$ to compute sets of ground terms for instantiating quantified variables. Unlike us, however, they do not calculate a solution upfront, but instead, propose a fair enumeration of the (least) solution of $\varDelta_F$ with certain properties. Assuming such enumeration, one can incrementally construct and check the quantifier-free formulas as needed\footnote{In practice, they guide the quantifier instantiation using model checking which, in turn, uses an SMT solver.}. If $\varDelta_F$ is \emph{stratified},   $F$ is in a decidable fragment, and  termination of the procedure is guaranteed. 
Otherwise the procedure can fall back on the quantifier engine of Z3 and provide helpful instantiation ground terms. 
Consequently, this technique can only act as an internal engine of an SMT solver and cannot provide a stand-alone formula simplification as ours does.

Variable expansion has also been proposed for quantified boolean formulas (QBF). In \cite{qbf_resolve_and_expand}, a reduction of QBF to propositional conjunctive normal form (CNF) is presented where universally quantified variables are eliminated via expansion. Similar to our approach, they  introduce cost functions, but with the goal of keeping the size of the generated CNF small. 
\section{Conclusion}
\label{sect:conclusion}
We described a general simplification approach for quantified SMT formulas.
Based on an analysis of the ground term occurrences at function applications, we compute  \emph{sufficient ground term sets} for each universally quantified variable. 
We proved that instantiating (thus eliminating) any variable whose computed set  is finite, results in an equisatisfiable formula. 
Elimination of each variable is independent of the others. Thus we  
improve the performance of our technique by restricting the set of eliminable variables: we defined a prioritization  algorithm that  
tries to maximize the number of eliminable variables while keeping the estimated elimination cost below a  threshold.  
We evaluated our approach using two configurations and two solvers on a large subset of the SMT-COMP benchmarks. 
Our results show that 
(1) SMT benchmarks contain  many variables that can be eliminated by our technique,
(2) our complete variable instantiation may introduce significant overhead and thus slow down the solvers,
(3)  instantiation along with  prioritization  shows  improvement of  the solving time and score. 

We believe that our technique can provide an easy framework for extending arbitrary SMT solvers with quantifier support. If we ignore termination and performance related rules when generating the set constraint system, we will have an incremental and fair procedure for building ground term sets. Using a finite model checker, like in \cite{mbqi}, can then provide a framework for extending SMT solvers with quantifier support. Investigating this idea is left for future work.



%
\label{sect:bib}
\bibliographystyle{plain}
\bibliography{ref}

\appendix

\section{Expansion Example}
\label{example}

The following example illustrates a case where eliminating one variable can result in increasing the occurrences of the other variables. This can introduce an overhead for the solver if the involved terms are complex.

\begin{example}
\label{examp1}
Let $\forall x \mid (\psi(x) \foor \forall y, z \mid \varphi(x,y,z))$ be the input formula, and $S_y =  \{gt_1, \ldots, gt_n\}$ be a set of sufficient ground
terms for the variable $y$. Suppose that the sets of sufficient ground terms of $x$ and $z$ are infinite. In this case, instantiating and eliminating $y$ will result in the formula 
\begin{gather*}\label{eq:examp1:form2}
  \forall x \mid (\psi(x) \foor \forall z \mid (\varphi(x,gt_1,z)
  \foand \ldots \foand \varphi(x,gt_n,z)))
\end{gather*}
which has a higher number of occurrences of the variables $x$ and $z$.   
\end{example}

\section{Proofs}
\label{a:sect:proofs}

%
\begin{repcorollary}{col:1}
If $\n{vGT}(x)_{\scs_A} \neq \infty$, then $\pi_x(v) \in M(\n{vGT}(x)_{\scs_A})$, for all $v \in |M|$.
\end{repcorollary}

\begin{proof}
The claim follows directly from the definition of $\pi_x$
\end{proof}

\begin{corollary}
\label{col:2}
For all $gt \in \n{Gr}(A)$, $M^{\pi_x}(gt) = M(gt)$.
\end{corollary}

\begin{proof}
By induction over the structure of $gt$. If $gt \in \n{Const}$, the claim follows directly from the definition of $M^{\pi_x}$. 
If, without loss of generality, $gt := f(t)$, where $f \in \n{Fun}$ and $t \in Gr$, we get by the induction hypothesis, $M^{\pi_x}(f(t)) = M^{\pi_x}(f)(M^{\pi_x}(t)) {\buildrel\rm i.h. \over=} M^{\pi_x}(f)(M(t))$. 
Now we have to distinguish between interpreted and uninterpreted functions. If $f$ is interpreted, the claim follows directly from the definition of $M^{\pi_x}$. 
If $f$ is uninterpreted, we get $M^{\pi_x}(f)(M(t)) = M(f)(\pi_x(f,1)(M(t)))$. Furthermore, we know, because of rule $R_2$ and $gt \in \n{Gr}(A)$, that $t \in \n{fGT}(f,1)_\sa$. 
Now we can use the definition of $\pi_x(f,1)$ and we get $\pi_x(f,1)(M(t)) = M(t)$.
\end{proof}

For a variable assignment $\beta$, a value $v \in |M|$ and a variable $x \in Var$, we use the notation $\beta_{x}^{v}$ to denote the modification of $\beta$ where $x$ is mapped to $v$.

%
\begin{replemma}{lemma:1}
Let $x$ be a variable with $\vgts{x} \neq \infty$, $\mc{M}$ a model, $\beta$ a variable assignment, and 
$\beta' = \lambda y.\ \n{if}\ \vgts{y} \dotsubset \vgts{x} \ \n{then}\ \pi_y(\beta(y))\ \n{else}\ \beta(y)$.
Then  
$(M, \beta') \models l$ implies $(M^{\pi_x}, \beta) \models l$
for all interpreted literals $l$ in $A$.
\end{replemma}

\begin{proof}
Because of the rules $R_4$ and $R_6$, without loss of generality, we can restrict $l$ to $l := op(x, gt_0)$ where $op \in \{=, <, \leq, >, \geq\}$ and $\beta'$ to $\beta'= \lambda y. \ \beta_{x}^{\pi_x(\beta(x))}(y)$. 
Let us now assume that $(M, \beta_{x}^{\pi_x(\beta(x))}) \models l$ and $(M^{\pi_x}, \beta) \not\models l$. 
For $op := \verb+"+<\verb+"+$, we get from rule $R_5$, $gt_0 \in \n{vGT(x)}_{\scs_A}$ and from the assumptions the inequality system $(\beta(x) \ge gt_0) \foand (\pi_x(\beta(x)) < gt_0)$, 
which implies that $|\beta(x) - \pi_x(\beta(x))|$ is not minimal, since $|\beta(x) - gt_0|$ is strictly smaller. 
For $op \in \{\leq, >, \geq\}$, the proof is similar to the previous case. For $op := \verb+"+=\verb+"+$, we get from rule $R_{13}$, $\{gt_0-1, gt_0+1\} \subseteq \n{vGT(x)}_{\scs_A}$ 
and from the assumptions, the inequality system $(\beta(x) \not = gt_0) \foand (\pi_x(\beta(x)) = gt_0)$, which is equivalent to $(\beta(x) \leq gt_0-1) \foor (gt_0+1 \leq \beta(x))) \foand (\pi_x(\beta(x)) = gt_0)$ 
and implies that $|\beta(x) - gt_0|$ is not minimal, since in the case $(\beta(x) \leq gt_0-1)$, $|\beta(x) - (gt_0-1)|$ is strictly smaller and in the case $(gt_0+1 \leq \beta(x))$, $|\beta(x) - (gt_0+1)|$ is strictly smaller.  
\end{proof}

Proposition \ref{proposition:1} provides a stronger result compared to lemma \ref{lemma:1}. It better reflects the intuition behind the rules $R_5$, $R_7$ to $R_{13}$. 
They guarantee that if a variable $x$ occurs as an argument of an interpreted operator, then there is at least one $gt_l \in \vgts{x}$ with $\not \models_{\mathcal{T}} l[gt_l/x]$. 
That is, $C[gt_l/x]$ is either valid or its satisfiability is determined by literals other than $l$. 
We proved lemma \ref{lemma:1} because it is sufficient for our main theorem, and it has a shorter proof.

\begin{proposition}
\label{proposition:1}
Let $C$ be a clause in $A$, $x$ a variable in $C$ with $\n{vGT(x)}_{\scs_A} \neq \infty$, and $M$ a model  of $C[\n{vGT(x)}_{\scs_A}/x]$, then either there exists an uninterpreted literal $l \in C$, where $M \models l[gt/x]$ for some $gt \in \n{vGT(x)}_{\scs_A}$, or there exists a (tautology) subclause $C'$ of $C$ whose literals are interpreted and $\models_{\mathcal{T}} C'$.
\end{proposition}

In the following, we use \emph{expressions} to refer to both terms and formulas. That is, $Expr = Term \cup For$.

%
\begin{replemma}{lemma:3}
Let $x$ be a variable with $\vgts{x} \neq \infty$, $\mc{M}$ a model, $\beta$ a variable assignment, and 
$\beta' = \lambda y.\ \n{if}\ \vgts{y} \dotsubset \vgts{x} \ \n{then}\ \pi_y(\beta(y))\ \n{else}\ \beta(y)$.
Then 
$
(M,\beta')(l)
= 
(M^{\pi_x}, \beta)(l) 
$ for all uninterpreted literals $l$ in $A$.
\end{replemma}

\begin{proof}
  To prove the claim, we show the statement $(M,\beta')(l) =
  (M^{\pi_x}, \beta)(l)$ for all expressions but variables $l \in
  \n{Expr \setminus Var}$ occurring in $A$ using structural induction.


  If $l$ is a ground term in $A$, then the claim follows directly from corollary \ref{col:2}.
  
  Let $l = f(t_{1:n})$ be a function application in $A$ with $f$ an uninterpreted function. 
  The evaluations of $l$ are
  \begin{align*}
	  (M^{\pi_x},\beta)(f(t_{1:n})) &= M^{\pi_x}(f)((M^{\pi_x}, \beta)(t_1), \dots, (M^{\pi_x}, \beta)(t_n))\\
    \quad &= M(f)(\pi_x(f,1)((M^{\pi_x}, \beta)(t_1)), \dots, \pi_x(f,n)(t_n))\\
    (M, \beta')(f(t_{1:n}) &= M(f)((M,\beta')(t_1), \dots, (M,\beta')(t_n))
  \end{align*}
  It suffices to show that 
  $\pi_x(f,i)((M^{\pi_x}, \beta)(t_i)) = 
  (M, \beta')(t_i)$ for $1 \leq i \leq n$.
  We do this by a case distinction over the type of the terms $t_i$.
  
  If $t_i = y$ is a variable with $\vgts{y} \ndotsubset \vgts{x}$, then $\beta'(y) = \beta(y)$. 
  Because of rule $R_1$ we additionally get $\fgt{f}{i} \ndotsubset \vgts{x}$, which implies that $\pi_x(f,i)$ is the identity.

  If $t_i = y$ is a variable with $\vgts{y} \dotsubset \vgts{x}$, then $\beta'(y) = \pi_y(\beta(y))$.
  Because of rule $R_1$ we get $\vgts{y} = \fgt{f}{i} \dotsubset \vgts{x}$, which implies that $\pi_x(f,i)= \pi_y$.

  If $t_i$ is a function application, we assume $t_i = s[x_{1:m}]$ for some term $s$. 
  By induction hypothesis,
  $\pi_x(f,i)((M^{\pi_x}, \beta)(s[x_{1:m}])) \eqih$ 
  $\pi_x(f,i)((M, \beta')(s[x_{1:m}]))$. 
  W.r.t. $\n{fGT}(f,i)_\sa$, there is two possible cases to consider:\\
  1) $\fgt{f}{i} \ndotsubset \vgts{x}$, then $\pi_x(f,i)$ is the identity and the claim follows directly. \\
  2) $\fgt{f}{i} \dotsubset \vgts{x}$, then because of rule $R_3$ $\vgts{x_i} \dotsubset \fgt{f}{i} \dotsubset \vgts{x}$, for all $1 \leq i \leq m$.
  This implies that $\beta'(x_i) = \pi_{x_i}(\beta(x_i))$ for all $1 \leq i \leq m$.
  Using this fact together with corollary \ref{col:1}, there exists for each $x_i$ a ground term $gt_i$, with $\pi_{x_i}(\beta(x_i)) = M(gt_i)$ and $gt_i \in \vgts{x_i}$.
  So we can write, $\pi_x(f,i)((M, \beta')(s[x_{1:m}])) = \pi_x(f,i)(M(s[gt_{1:m}]))$.
  Because of rule $R_3$ we know that $s[gt_{1:m}] \in \fgt{f}{i}$, and so $M(s[gt_{1:m}]) \in M(\fgt{f}{i})$. 
  Finally the claim follows from the definition of $\pi_x(f,i)$ for values in $M(\n{fGT}(f,i)_\sa)$ and the assumption that $\fgt{f}{i} \dotsubset \vgts{x}$.

  Let $l := f(t_{1:n})$ be an expression with $f$ an interpreted function. 
  Using the definition of $M^{\pi_x}$ for interpreted functions,
  $(M^{\pi_x}, \beta)(f(t_{1:n})) =$ 
  $M(f)((M^{\pi_x},\beta)(t_1), \dots, (M^{\pi_x},\beta)(t_n))$. 
  Since $l$ is uninterpreted, all $t_i$s are non-variables and we can use the induction hypotheses on them and get,
  $M(f)((M^{\pi_x},\beta)(t_1), \dots, (M^{\pi_x},\beta)(t_n)) =$
  $M(f)((M,\beta')(t_1), \dots, (M,\beta')(t_n))$.
  Now the claim follows directly from the definition of $M$.

\end{proof}

\end{document}